\documentclass{article}

\usepackage[utf8]{inputenc}
\usepackage{hyperref}
\usepackage{adjustbox}
\usepackage{url}
\usepackage{booktabs}
\usepackage{amsfonts}
\usepackage{microtype}
\usepackage{tikz}
\usepackage{amsthm}
\usepackage{amsmath}
\usepackage{algorithm}
\usepackage{algorithmic}
\usepackage{graphicx}
\usepackage{subcaption}
\DeclareMathOperator*{\argmax}{arg\,max}
\theoremstyle{definition}
\newtheorem{proposition}{Proposition}
\newtheorem{theorem}{Theorem}

\newtheorem{definition}{Definition}

\usepackage[final]{dai_2019} 

\begin{document}

\title{On Maximizing Egalitarian Value in K-coalitional Hedonic Games
}


 \author{
   Naftali Waxman
   \\ 
   Bar-Ilan University\\ 
   Israel \\
   \texttt{vaxmann@cs.biu.ac.il} \\
   \And
   Sarit Kraus
   \\ 
   Bar-Ilan University\\ 
   Israel \\
   \texttt{sarit@cs.biu.ac.il} \\
   \And
   Noam Hazon
   \\ 
   Ariel University\\ 
   Israel \\
   \texttt{noamh@ariel.ac.il} \\
 }

\maketitle 

\begin{abstract}
This paper considers the problem of dividing agents among coalitions. We concentrate on Additively Separable Hedonic Games (ASHG's), in which each agent has a non-negative value for every other agent and her utility is the sum of the values she assigns to the members of her coalition. 
Unlike previous work, we analyze a model where exactly $k$ coalitions must be formed, and the goal is to maximize the utility of the agent which is worst off, i.e., the egalitarian social welfare. We show that this problem is hard, even when the number of agents should be equally divided among the coalitions. We thus propose a heuristic that maximizes the egalitarian social welfare and maximizes the average utility of each agent as a secondary goal. Using extensive simulations, both on synthetic and real data, we demonstrate the effectiveness of our approach. Specifically, our heuristic provides solutions that are much fairer than the solutions that maximize the average social welfare, while still providing a relatively high average social welfare.
\end{abstract}

\section{Introduction}

Coalition formation is one of the fundamental research problems
in multi-agent systems \cite{chalkiadakis2011computational}. 
Broadly speaking, coalition formation is concerned with partitioning
a population of agents into disjoint teams (or coalitions) with the aim that some
system-wide performance measure is maximized or that the selected partition
is stable according to some stability measure.

One of the main classes of coalition formation games is Hedonic Games
\cite{bogomolnaia2002stability, banerjee2001core}. In these games, each agent’s 
utility solely depends on the other agents within her coalition and not
on how other agents are partitioned. There is an important subclass of Hedonic Games, which is called
Additively Separable Hedonic Games (ASHG's) \cite{aziz2013computing, aziz2011stable}. 
In these games each agent has a value for each of the other agents, and her utility
in a given coalition formation is the sum of the values she assigns to the members of
her coalition. For example, consider the problem of dividing
students into classes. In this case each student may specify her friends, and her utility is the number of friends she has within the class to which she is assigned. Another example is of assigning agents to complete tasks. Each agent can specify how well she works with other agents, and her utility is the sum of the values she specified to the members of her coalition.

The usual objective in ASHG's is to maximize the sum of individual utilities of the agents, i.e., the utilitarian social welfare. In addition, one assumption made in most coalition formation games is that any number of coalitions can be formed. 
However, it is sometimes more important to maximize the utility of the agent that is worst off, i.e. the egalitarian social welfare, and in many real-world problems exactly $k$ coalitions are required. For example, when dividing students into classes there is a known number of classes that have to be formed, and this number cannot be modified. There is no way to form more or less classes than agreed upon by the school and the administration. Moreover, a partition which yields an average of $3$ friends per student but leaves a couple with no friends at all would be considered by many to be worse than a partition with an average of $2.5$ friends per student in which everyone has at least one friend.


In this paper we study ASHG's with non-negative utilities, where the objective is to maximize the egalitarian social welfare while exactly $k$ coalitions are allowed to be formed. We first show that maximizing the egalitarian social welfare is hard even for fixed $k$ and equally sized coalitions. However, for simple games (see the definition below), finding a coalition structure with an egalitarian value of $1$ can be done in polynomial time, if one exists. Similarly, we show that finding a coalition structure with egalitarian value of $2$ can be done in polynomial time, if the game is also symmetric and has a specific structure.
Unfortunately, we show that when maximizing the egalitarian social welfare the loss in the utilitarian social welfare is unbounded.

As noted before, our problem is well-motivated by real world scenarios. 
We thus propose a simulated annealing heuristic that maximizes the egalitarian social welfare. We also propose a variant of hill climbing, denoted \emph{LexiClimb}, which maximizes the the egalitarian social welfare and the average utility of each agent as a secondary goal. 
For evaluating the heuristics we used real preferences of students from $3$ schools as well as synthetic data, and our extensive simulations demonstrated the effectiveness of LexiClimb. Specifically, LexiClimb provides solutions that are much fairer than the solutions that maximize the average social welfare, while still providing a relatively high average social welfare.


\section{Related Work}
There are few works that consider the objective of maximizing the egalitarian social welfare of a coalition structure. Skibski et al.~\cite{skibski2016k} propose to study the egalitarian social welfare in general non-transferable utility games. 
In the domain of ASHG's, Peters \cite{peters2016graphical} showed that finding a maximum egalitarian partition is polynomial time tractable, if the game is symmetric and its graphical representation is of bounded tree width. Aziz et al.~\cite{aziz2013computing} considered arbitrary utilities, and showed that 
finding a maximum egalitarian partition is NP-hard, 
and verifying the existence of such partition is coNP-complete. All of these works assume that every coalition structure is feasible, while we assume that exactly $k$ coalitions are allowed to be formed. 
%

The restriction on the number of coalitions has been mostly overlooked. Indeed, Skibski et al.~\cite{skibski2016k} study $k$-coalitional cooperative games under  the transferable utility setting, and develop a dedicated extension of the Shapley value for this game.
Sless et al.~\cite{sless2014forming, sless2018forming} initiated the study of ASHG's where exactly $k$ coalitions must be formed.
However, their goal is to to maximize the utilitarian
social welfare.

When looking at the decision problem, whether a coalition formation of egalitarian value of at least $m$ 
exists, there are some tractable instances or guaranteeing constraints. For every ASHG there is an equivalent representation as a weighted graph, and finding a coalition structure is therefore equivalent to finding partition of the graph. Hence, a notable result by Stiebitz
\cite{stiebitz1996decomposing} shows that every simple symmetric ASHG can be partitioned to a high egalitarian value coalition structure if the original minimum degree of its graphical representation is high enough. However, this result is not applicable for ASHG's which are not symmetric \cite{alon2006splitting}. 
Bang et al.~\cite{bang2019undirected} proved that in the general case deciding whether a simple symmetric ASHG has a 2-coalitions coalition structure with egalitarian values of at least $(\delta_1, \delta_2)$ is generally hard, except for some special cases regarding $(\delta_1, \delta_2)$ and the graphical representation's minimum degree. For asymmetric case the same problem was proved to be polynomial-time solvable by \cite{bang2016finding} for the case of $\delta_1 = \delta_2 = 1$ and was proven to be NP-hard recently by \cite{bang2018directed}for any higher values.

We later focus on heuristics aiming at maximizing the egalitarian social welfare of a $k$-coalitions coalition structure of an ASHG, while constraining the sizes of each coalition, as well as maximizing the utilitarian social welfare as a secondary goal. The problem of maximum utilitarian partition of roughly equally sized partitions is a known problem called Graph Partitioning \cite{bulucc2016recent} and is proven to be NP-hard, while having very good practical 
algorithms.

Another related problem with a similar goal is the Satisfactory Partition Problem \cite{shafique2002satisfactory}. In this problem the goal is to find a 2-partition of a graph where each vertex has at least as many neighbours in his subgraph as in the other subgraph. A generalization of this problem for k-partitions was introduced by \cite{bazgan2010satisfactory}, where they proved hardness for the different variations of the problem.

\section{Preliminary Definitions}
A hedonic coalition formation game is given by a tuple $(N, \succeq)$ where $N = \{ 1, \dots n \}$ is a finite, non-empty set of players and $\succeq = (\succeq_1, \dots, \succeq_n)$ is a preference profile, specifying for each player $i \in N$ his relation $\succeq_i$ over the set $\mathcal{N}_i = \{ C \subseteq N | i \in C \}$. We say that player $i$ prefers (strictly prefers) coalition $C$ over coalition $C^{'}$ if it holds that $C \succeq_i C^{'} \  (C \succ_i C^{'})$. A solution for a hedonic game is a partition $\pi$ of $N$, also called a coalition structure (CS). We will use the notion $\pi (i)$ to denote the coalition in $\pi$ that includes player $i$, and $P$ to refer to all possible partitions.

$K$-coalitional games are games where exactly $K$ coalitions must be formed.

An additively separable hedonic game (ASHG), $(N, \succeq)$, is a hedonic game where each player $i$ has a value $v_i (j)$ for each other player $j$ being in the same coalition as hers, and the utility of agent $i$ being in a coalition $C \in \mathcal{N}_i$ is the sum $v_C (i) = \sum_{j \in S\setminus\{i\}} v_i (j)$. We say an ASHG is simple if for any two players $i,j$ it holds that $v_i (j) \in \{0,1\}$, and is symmetric if $v_i(j) = v_j(i)$.

The utilitarian social welfare of a partition $\pi$ is defined as the sum of individual utilities of all of the agents. In this paper we will often use the average utility of all players as an indicator for the utilitarian social welfare as they are linearly correlated, and in our work we want to put the perspective on the individual.
The egalitarian social welfare of a partition is given by the utility of the agent that is worst off. Formally we have: 
$$U(\pi) = \sum_{i\in N} v_\pi (i) / n$$ 
$$E(\pi) = \min \{v_\pi (i) | i\in N\}$$ 

A maximum egalitarian coalition structure is a coalition structure that maximizes the egalitarian social welfare compared to all other possible coalition structures. Formally: 
$$\argmax_{\pi \in P}{E(\pi)} $$
The maximum egalitarian $k$-coalition CS is defined similarly, as a CS that maximizes the egalitarian social welfare among all other $k$-coalition CS's.

We define the maximum utilitarian ($k$-coalition) coalition structure the same way except with the utilitarian social welfare.

A $(m_{1},m_{2},\dots,m_{k})$-partition is a $k$-partition where each coalition $C_{1},C_{2},\dots,C_{k}$ is at least of size $m_{1},m_{2},\dots,m_{k}$ respectively. An equal sized $k$-partition is a partition into $k$ coalitions where each is of size $n/k$. We will look at a case where roughly equally sized coalitions are required, and for that we use the same definition from Graph Partition: A $(k,1+\epsilon)$-partition is a partition into $k$ coalitions where each coalition is at maximum of size $(n/k)*(1+e)$.

The main problem we discuss in the paper is finding the coalition structure that maximizes the egalitarian value. Formally we define the decision problem as follows: 
\begin{definition}{\textsc{DecisionEgalitarianCS}:}
Given an additively separable hedonic game $G = (N, \succeq)$, and positive integers $k, \delta$. 
Decide whether there exists a $k$-CS of $G$ of egalitarian value of at least $\delta$.
\end{definition}
The equal variation \textsc{DecisionEgalitarianEqualCS} is the same problem, 
but with the restriction that we look only for equal sized $k$ coalition structures.

ASHGs can be represented as graphs, hence we use notions of graphs in this paper in many cases instead of
those of game theory; Vertices and edges represent the players and their utility functions respectively. 
A $k$-coalition structure is equivalent to a $k$-cut, and the total weight and the minimum weight are the
same as the utilitarian social welfare and the egalitarian social welfare, respectively. For the minimum out-weight of a digraph
$D = (V,E)$ we will use the notion $\delta(D)$ and the total weight is simply $|E|$. In the case of simple
ASHGs we can use the degrees instead of the weights, and in the case of symmetric ASHGs an undirected graph
instead of a graph. Most of the other definitions also follow through very clearly.

In this paper we will talk only on digraphs, unless specifically stated otherwise.
\section{Theoretical remarks and observations}

\subsection{Equally Sized Coalitions}
Bang et al \cite{bang2018directed} proved that \textsc{DecisionEgalitarianCS} $\delta \geq 2$. We prove that the equal sized version of the problem is also hard:
\begin{theorem}
\textsc{DecisionEgalitarianEqualCS} is NP-complete.
\end{theorem}
\begin{proof}
We use a reduction from the general case. We show that if there is a polynomial-time algorithm $A^{'}$ that solves \textsc{DecisionEgalitarianEqualCS}, there is also a polynomial-time algorithm $A$ that solves \textsc{DecisionEgalitarianCS}.

Let us assume there exists such an algorithm $A^{'}$. Let $G$ be a simple ASHG with $n$ players and let $G, k, \delta$ be an input for a \textsc{DecisionEgalitarianCS} problem. We use the following algorithm $A$ to solve it: 
\begin{enumerate}
  \setcounter{enumi}{0}
  \item Create a game $G^{'}$ which is $G$ with $n(k-1)$ new players, each of which has all original players as neighbours.
  \item Run $A$ on $G^{'}, k, \delta$, and returns the answer of $A$.
\end{enumerate}
First let's assume that $A$ returns \textit{yes} on $G^{'}, k, \delta$. 
Since all new players' utilities are from the original players,
original players must be partitioned among all the coalitions in the coalition structure of $A$.
We can take the same coalition structure, remove all new players and get a coalition structure of
exactly $k$ coalitions. Also, the original players' utilities haven't changed, so by removing the
new players we haven't changed their utilities. Hence the egalitarian value is still at least $\delta$.
Secondly, let's assume that there is a $k$ - coalition structure in the original
game of egalitarian value of at least $\delta$, then $A$ returns \textit{yes}.
We can simply add players to any of the current coalitions, up to $n$ players in each coalition.
The new coalition structure is obviously of exactly $k$ equally sized coalitions. 
If we run this algorithm we get the maximum egalitarian 2-CS: Let's suppose the maximum egalitarian value possible is $m$ and that $A^{'}$ returns $m^{'}$. Note that $m$ cannot be bigger than $n/2$ as at least 1 of the coalitions is at most of size $n/2$. We can take a 2-partition with an egalitarian value $m$ and add $n-2$ players in such a way as to obtain 2 equally sized coalitions of $n-1$ players each. Since the new $n-2$ players have everyone as friends they each have a utility of $n-2$, and the egalitarian value stays $m$. Now suppose $A^{'}$ returns $m^{'}$, then we can take away the $n-2$ new players and still have $m^{'}$ as an egalitarian value. This is true because, just like for $m$, there is a coalition in $G^{'}$ with at most $n/2$ from $G$, hence there is a player with a utility of at most $n/2-1$, which is smaller than $n-2$. So we get $m \geq m^{'}$, and with the result from above $m = m^{'}$.
\end{proof}

\subsection{Egalitarian Value of at Least 1}
In 2016 \cite{bang2016finding} proved that checking whether a digraph has a partition into 2 subgraphs of sizes at least $(m_1,m_2)$ such that each of them has an out-degree of  least 1 can be done in polynomial time if $(m_1,m_2)$ are fixed. We prove that the same can be done for a $k$-partition of fixed sizes $m_1, m_2, \dots, m_k$ for a fixed $k$.

\begin{theorem}
Let $D$ be a subdigraph, $k>2$, and let $(m_1, m_2, \dots, m_k)$ be positive integers. Deciding whether $D$ has a $(m_1, m_2, \dots, m_k)$-partition with out-degree of at least 1 is polynomial-time solvable for a fixed $k$.
\end{theorem}

\begin{proof}
We show an inductive algorithm which solves this problem that runs in polynomial time. For the case of $k=2$ we use the fact that a graph satisfies the requirement if and only if one of the follow holds \cite{bang2016finding}:
\begin{enumerate}
  \setcounter{enumi}{0}
  \item $D$ has two disjoint directed cycles $C_1$ and $C_2$ of length at least $m_1$ and $m_2$ respectively.
  \item $D$ has a subdigraph $D^{'}_{1}$ such that $\delta (D^{'}_{1}) \geq 1,\  m_1 \leq |D^{'}_{1}| \leq 2m_1 -2$, and $D - D^{'}_{1}$ has a directed cycle of length at least $m_2$.
  \item $D$ has a subdigraph $D^{'}_{2}$ such that $\delta (D^{'}_{2}) \geq 1,\  m_2 \leq |D^{'}_{2}| \leq 2m_2 -2$, and $D - D^{'}_{2}$ has a directed cycle of length at least $m_1$.
  \item $D$ has two disjoint subdigraphs $D^{'}_{1}, D^{'}_{2}$ such that $\delta (D^{'}_{i}) \geq 1,\  m_i \leq |D^{'}_{i}| \leq 2m_i -2$ for $i = 1,2$.
\end{enumerate}
Each of the above can be checked in polynomial time.
In short, (1) is true because checking if a disjoint union of directed cycles are a subdigraph of a given digraph
can be done in polynomial time as it is a planar graph with no in-degree or out-degree of 2 or more.
(2) is true because checking whether a digraph $S$ of a fixed size has an out-degree at least 1 can be done in constant time, there are polynomial possible sets of size $m_1$, and checking the existence of a directed cycle in $D-S$ can be done in polynomial time as mentioned above. (3) and (4) follow simply from the reasons described about (2).

We prove $k>2$ with induction. Assume that the algorithm works in polynomial time for every $2 \leq k^{'} < k$, then a subdigraph $D$ has a $(m_1, m_2, \dots, m_k)$-partition if and only if:
\begin{enumerate}
  \setcounter{enumi}{0}
  \item $D$ has $k$ disjoint cycles $C_1, C_2, \dots, C_k$ each of size at least $m_1, m_2, \dots, m_k$ respectively.
  \item $D$ has a subdigraph $D^{'}_{i}$ such that $\delta (D^{'}_{i}) \geq 1,\  m_i \leq |D^{'}_{i}| \leq 2m_i -2$, and $D- D^{'}_{i}$ has a $(m_1, m_2, \dots, m_{i-1}, m_{i+1}, \dots, m_k)$-partition with out-degree of at least 1 for some $1\leq i\leq n$.
\end{enumerate}

Both statements can be checked in polynomial time: 
(1) For each combination of $m^{'}_{1},m^{'}_{2}, \dots,m^{'}_{k}$ such that $m_{i} \leq m^{'}_{i} \leq n$ and $\sum_{i=1}^{k} m^{'}_{i} \leq n$, we check if the union of $k$ disjoint cycles $C_{1}, C_{2}, \dots, C_{k}$ of sizes $m^{'}_{1},m^{'}_{2}, \dots,m^{'}_{k}$ is a subdigraph of $D$. There are less than $n^{k}$ possibilities for $m^{'}_{1},m^{'}_{2}, \dots,m^{'}_{k}$, and for each of these possibilities the check can be done in polynomial time as stated above.
(2) follows the same statements from the 2-partition case. For each $1\leq i\leq n$ there is a polynomial number of subdigraphs of size $m_{i} \leq |D^{'}_{i}| \leq 2m_i -2$. For each subdigraph $D^{'}_{i}$, checking if $\delta(D^{'}_{i})$ can be done in a fixed time, and checking if $(D-D^{'}_{i})$ has $(m_1, m_2, \dots, m_{i-1}, m_{i+1}, \dots, m_k)$-partition with out-degree of at least 1 is done in polynomial time by the induction assumption.
\end{proof}

\subsection{Egalitarian value of 2 in Symmetric Games}
We take a small break to talk about symmetric games and their matching representation - undirected graphs - in this section, as there are not many results on this subject.
For undirected graphs of minimum degree at least $k + 1$ we can check in polynomial time whether it has a $k$-partition of minimum degree at least 2 or not as follows:
First, check if the graph contains $k$ vertex-disjoint cycles, which can be done in polynomial time \cite{bodlaender1994disjoint}. 
If not, then the answer to the problem is also no, as every graph of minimum degree 2 has a cycle in it, and a $k$-partition of minimum degree at least 2 guarantees at least $k$ disjoint cycles. 
If the graph has $k$ disjoint cycles, the answer to the problem is also yes. We start by partitioning the graphs into the $k$ disjoint cycles and an `leftovers`. Any vertex which is not in one of the cycles is added to any one of them which contains at least 2 of its neighbours, arbitrarily. If after this process there are still some nodes left outside, it holds that each one of them either has at least 2 neighbours among them or 1 neighbour in each part, as well as the `leftovers`. In that case we take all of the `leftovers` and add them to one of the parts arbitrarily, as in the union each one of the `leftovers` is guaranteed to have at least 2 neighbours.

\subsection{No bound available}

We would like to bound the utilitarian social welfare of the maximum egalitarian social welfare $k$-CS to the maximum utilitarian social welfare possible. We show that this is not possible with an example of the worst case ratio. Optimally, the best utilitarian social welfare possible is $|E|$. It is clear that if we can obtain a $k$-CS with egalitarian value of at least 1, then the lower bound of the utilitarian social welfare is $n$, and if we cannot there is no reason to compare the two. 
We present an example where there is a $k$-cut with minimum degree 1 and total degree of exactly $n$, whereas the maximum total degree possible approaches $|E|$ the larger the graph: Let $n$ be an integer divisible by $k$, and denote $n/k$ as $m$. Let $D$ be a digraph with $n$ vertices numbered $a_1$ to $a_n$, where every vertex has $k$ out-edges pointing at the $k$ next vertices (i.e. vertex $a_1$ pointing at $a_2$ to $a_{k+1}$ and vertex $a_n$ pointing at $a_1$ to $a_k$). It is possible to obtain minimum degree $1$ by simply dividing the digraph into $k$ disjoint cycles $V_i = \{ a_i, a_{i + k}, a_{i + 2k}, \dots, a_{i + m*k}\}, \  i \in \{ 1, \dots, k\}$. We prove that this is the only $k$-cut that provides a minimum out-degree $1$. 

\begin{proposition}
Any subgraph in the above settings with minimum degree of at least 1 is of size at least $m$.
\end{proposition}
\begin{proof}
For the proof we will use the fact that every digraph with a minimum out-degree of at least 1 has a cycle in it. Let $D^{'}$ be a subgraph with less than $m$ vertices. For aesthetic reasons we assume $a_{1}$ is in $D^{'}$. In order to close a cycle, a vertex in the $k$ prior vertices to $a_{1}$ must be in the subgraph as well, i.e. one of $a_{n - k +1}, a_{n - k +2}, \dots, a_{n}$. Let us assume by contradiction that $D^{'}$ has a minimum degree of at least 1. Every vertex can reach the next $k$ vertices. Since we have at most $m-1$ vertices, and $m-2$ vertices other than the last vertex, we can reach at most $(m - 2)*k = n - 2k$ vertices after $a_{1}$. Hence the last vertex we can reach is $a_{n-2k + 1}$, which is too far from $a_{1}$ to close a cycle. Since every subgraph is of a size at least $m$ it holds that all subgraphs are of size exactly $m$. 
\end{proof}
Now we assume by contradiction that there are 2 vertices in the subgraph with a difference smaller than $k$. We follow the same arguments from the proof above, but now we have $m$ vertices. We know that 2 vertices are closer than $k$, so we gain the highest index from above plus less than $k$. At most this is $1 + n - 2k + (k - 1) = n - k$ which is still too far from $a_1$.

For a min $k$-cut we partition the graph into $m$ adjacent vertices: $V_i = \{ a_{i*k + 1}, a_{i*k + 2}, \dots, a_{i*k + m} \}$. In this cut every subgraph has $m-k$ vertices with all of their edges in the subgraph, and the last vertices have $k-1, k-2, \dots, 1, 0$. So we have a total degree of $k* [(m-k)k + (k-1)*k/2$. Even if we take out the $(k-1)*k/2$ part, we get to $k*(n-k^2) = n*k - k^3$. For every $\epsilon > 0$ there is no guarantee that $n*k-k^3 < k*(1+\epsilon)$.

\section{Heuristics and Experiments}
\subsection{Algorithms}

\begin{algorithm}
\caption{Simulated Annealing(CS)}\label{euclid1}
\begin{algorithmic}[1]
\STATE $\textit{BestCS} \gets \text{CS}, \textit{BestUtil} \gets \textit{util(CS)}$
\WHILE {$\textit{step < stepsLimist}$}
\STATE $\textit{SwapOrMove} = \textit{Random(\{Move,Swap\})}$
\IF {$\textit{SwapOrMove is Move} $}
\STATE $\textit{C1} \gets \textit{Random(CS) s.t. C1>MinSize}$
\STATE $\textit{C2} \gets \textit{Random(CS) s.t. C2<MaxSize}$
\STATE $a \gets \textit{Random(C1)}$
\STATE $C1 \gets C1\setminus \{a\}, C2 \gets C2\cup\{a\}$
\ELSIF{$\textit{SwapOrMove is Swap} $}
\STATE $\textit{C1, C2} \gets \textit{Random(CS)}$
\STATE $a \gets \textit{Random(C1)}, b \gets \textit{Random(C2)}$
\STATE $C1 \gets C1\setminus \{a\}\cup\{b\}, C2 \gets C2\setminus \{b\}\cup\{a\}$
\ENDIF
\STATE $\textit{CurrUtil} \gets \textit{util(CS)}$
\IF {$Accept(BestUtil, CurrUtil, temp) > Rand(0,1)$}
\STATE $\textit{BestCS} \gets \textit{CS}$
\ELSE 
\STATE $\textit{Revert(CS)}$
\ENDIF
\ENDWHILE
\RETURN {$bestCS$}
\end{algorithmic}
\end{algorithm}

For practical results, we show an algorithm that aims at maximizing the egalitarian social welfare while maintaining the utilitarian social welfare as high as possible. For the calculation of the optimal utilitarian value we use KaHIP, a graph partition solver, as it is considered one of the best current solvers \cite{bulucc2016recent}. 
We started by using a simulated annealing algorithm to try and maximize the egalitarian value. We used the utility function of $n * \delta - n_{\delta}$ where $n$ is the number of vertices, $\delta$ the current minimum out-degree and $n_{\delta}$ is the number of vertices of out-degree $\delta$, the temperature $0.8step/steplimit$ and the acceptance probability being $e^{(curr\_util - old\_util)/temp} $. The utility function ensures that a partition with a high minimum out-degree will have a higher value than one with a lower minimum out-degree, and also tries to minimize the number of vertices of the current minimum out-degree.
In order to improve the results obtained by the simulated annealing, we used the concept of leximin. It means that we try to minimize the number of people in the worst situation, then go to the next level and try to minimize that as well. Formally leximin is a lexicon order that works as follows: we say $(a_{1}, a_{2}, \dots, a_{n}) >_{leximin} (b_{1}, b_{2}, \dots, b_{n}) \iff a_{i} < b_{i} \ for\  the\  first\  i\  where\  a_{i}\neq b_{i}$. In our case it serves our purpose to take care of the worst students first. We compare two partitions by the leximin order of their list of out-degrees. The algorithm starts with a given partition and moves vertices between the subgraphs each iteration by local improvements: Look at two random coalitions and move the vertex which improves the leximin order the most. After we find no improvement for several iterations in a row we stop and start from another coalition structure. After a number of iterations 
decided preemptively, the algorithm outputs the best partition found.

\begin{algorithm}
\caption{Leximin(CS)}\label{euclid}
\begin{algorithmic}[1]
\STATE $\textit{BestCS} \gets \text{CS}$
\STATE $\textit{NoImprovementCounter} \gets \textit{0}$
\WHILE {$\textit{NoImprovementCounter} \leq \textit{Limit}$}
\STATE $improvementFlag \gets false$
\STATE $\textit{C1} \gets \textit{Random(CS) s.t. C1>MinSize}$
\STATE $\textit{C2} \gets \textit{Random(CS) s.t. C2<MaxSize}$
\FOR {$a \in \textit{C1}$}
\STATE $currCS \gets CS\setminus \{C1,C2\}\cup\{C1\setminus\{a\},C2\cup\{a\}\}$
\IF {$lex(currCS) > lex(bestCS)$}
\STATE $bestCS \gets currcS$
\STATE $improvementFlag \gets true$
\ENDIF
\ENDFOR
\IF {$improvementFlag$}
\STATE $NoImprovementCounter \gets 0$
\ELSE
\STATE $NoImprovementCounter++$
\ENDIF
\ENDWHILE
\RETURN {$bestCS$}
\end{algorithmic}
\end{algorithm}

\subsection{Results}
For the experiments we used real data from three schools as well as randomly generated graphs. In our experiments we checked both unweighted and weighted graphs using Borda-like weights for the preferences of the students, where a student's most preferred friend gets a weight of $m$ where $m$ is the number of friends allowed to be chosen, the next most preferred gets $m-1$ and so on. Here we focus on few of the results while there are more in the appendix showing the same tendency between algorithms.

Data from one school consisted of 127 students with 3 friends each. After running all of the algorithms on the graph as unweighted, we obtained the results shown in Table 1. As can be seen in the table, there's a big difference between the average degree obtained by KaHIP and by simply running Simulated Annealing or Leximin. As such, we tried running the Simulated Annealing and Leximin from different partitions obtained by KaHIP and got much better results. We can see that for the price of only one friendship (320 compared to 319) we can change the partition to ensure that every student gets at least one of her friends, which is a crucial difference.
The large difference in performance between running Simulated Annealing and Leximin from randomly generated partitions and from ones resulted by KaHIP was consistent through all experiments. This can be seen more clearly in Table 2 and Figure 1. In Table 2 we see the results of a second school which consists of 146 students. Each can rank up to five friends, with each friend's weight matching the friend's ranking. In Figure 1 we see the results on synthetically generated graphs with the same settings; a weighted graph where each node has five neighbours, uniformly generated. Each column represents the average of 100 random graphs, all of the same size matching the column. We see that as the graph is small the Simulated Annealing slightly outperforms the Leximin when it comes to the minimum weight, but the larger it is the worse it gets, while the Leximin is consistent. When it comes to the average weight KaHIP always dominates the other algorithms as expected, but the Leximin from KaHIP isn't so far behind, while it obviously by far outperforms the KaHIP's minimum weight performance. The surprising result is that by starting with KaHIP, the other algorithms not only improve their utilitarian value, but also their egalitarian value, as remarkably seen in Table 2.

\begin{table}
\begin{adjustbox}{width=0.48\textwidth}
\begin{tabular}{|l|l|l|l|l|}
\hline
                 & \textbf{Min Deg} & \textbf{Avg Deg} & \textbf{Total Deg} &\textbf{Gini} \\ \hline
\textbf{KaHIP}   & 0                & 2.52    &320     & 0.113                \\ \hline
\textbf{SA}      & 1                & 1.65   &210     & 0.133                \\ \hline
\textbf{SA\_KH}  & 1                & 2.38    &303     & 0.095                \\ \hline
\textbf{Lex}     & 1                & 1.85    &235   & 0.105                \\ \hline
\textbf{Lex\_KH} & 1                & 2.51    &319     & 0.096                \\ \hline
\end{tabular}
\end{adjustbox}

\caption{School 1 Results}
\begin{adjustbox}{width=0.48\textwidth}
\begin{tabular}{|l|l|l|l|l|}
\hline
                 & \textbf{Min Weight}& \textbf{Avg Weight} & \textbf{Total Weight} &\textbf{Gini} \\ \hline
\textbf{KaHIP}   & 0                & 11.50             & 1679  &0.084        \\ \hline
\textbf{SA}      & 2                & 5.0              & 730    &0.101           \\ \hline
\textbf{SA\_KH}  & 4                & 10.67             & 1558   &0.089             \\ \hline
\textbf{Lex}     & 1                & 7.26              & 1060  &0.100              \\ \hline
\textbf{Lex\_KH} & 4                & 11.27             & 1630  &0.076              \\ \hline
\end{tabular}
\end{adjustbox}

\caption{School 2 Results}
\label{my-label}
\end{table}

\begin{figure}
  \centering
  \begin{subfigure}[b]{0.8\linewidth}
    \includegraphics[width=\linewidth]{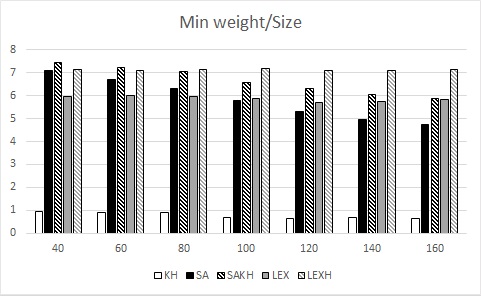}

  \end{subfigure}
  \begin{subfigure}[b]{0.8\linewidth}
    \includegraphics[width=\linewidth]{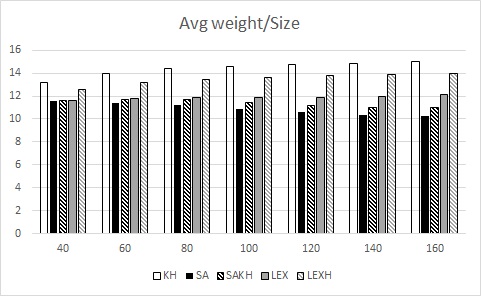}

  \end{subfigure}
  \begin{subfigure}[b]{0.8\linewidth}
    \includegraphics[width=\linewidth]{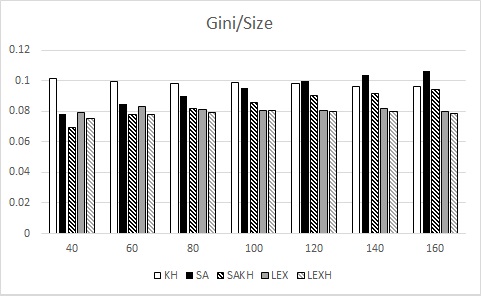}

  \end{subfigure}
  \caption{Performance of algorithms over graphs of different sizes. (a) Minimum Weight (b) Average Weight (c) Gini Coefficient}
  \label{fig:weighted size}
\end{figure}
As a measure of fairness we use the Gini Coefficient, a measure of statistical dispersion which in most cases represents the measurement of inequality when distributing wealth among agents. For our problem we can use it to estimate how fair a certain coalition structure is: the lower the Gini Coefficient, the more egalitarian and fair a distribution is. We see that Leximin consistently provides a low Gini Coefficient and as such results in fair coalition structures with minimal differences between students. It so happens that Leximin almost constantly results in the lowest Gini Coefficient when initiated with a KaHIP partition. We also see that when the Simulated Annealing performs very well it has a lower Gini than the Leximin, as we see in the case of graphs of size 40. On the contrary, when it performs poorly it also results in unfair coalitions, even worse than KaHIP, as we see in the case of graphs of size 140 or 160.

\section{Conclusion}
We analyzed the problem of coalition formation where a fixed number of equally sized coalitions must be formed with the goal of maximizing the egalitarian social welfare. We showed that the problem cannot be done in polynomial time, even when only 2 coalitions must be formed. We provided a polynomial time algorithm in the case of a fixed $k$ and checking for an egalitarian value of at least 1, and a similar result for an egalitarian value of at least 2 in the case of symmetric games. With respect to the utilitarian social welfare we provided some heuristics that try to maximize it as a secondary goal to the egalitarian social welfare. We showed that a bound between the utilitarian social welfare of the maximum egalitarian social welfare coalition structure and the maximum utilitarian possible does not exist, except for the trivial one.
We then showed the performance of three algorithms on real data and synthetic data. KaHIP had the best results with respect to the utilitarian value, as expected, but it ignored the egalitarian value. Simulated Annealing and Leximin had better results on the egalitarian value, and when initiated with partitions generated by KaHIP had even better ones, resulting in a minor loss to the utilitarian value and a significant gain to the egalitarian value.
Identifying a family of graphs on which the problems stated above are tractable or bounded by some factors is the subject of further research.

\clearpage
\small
\bibliographystyle{named}
\bibliography{dai_2019}

\begin{thebibliography}{}

\bibitem[\protect\citeauthoryear{Alon}{2006}]{alon2006splitting}
Noga Alon.
\newblock Splitting digraphs.
\newblock {\em Combinatorics, Probability and Computing}, 15(6):933--937, 2006.

\bibitem[\protect\citeauthoryear{Aziz \bgroup \em et al.\egroup
  }{2011}]{aziz2011stable}
Haris Aziz, Felix Brandt, and Hans~Georg Seedig.
\newblock Stable partitions in additively separable hedonic games.
\newblock In {\em The 10th International Conference on Autonomous Agents and
  Multiagent Systems-Volume 1}, pages 183--190. International Foundation for
  Autonomous Agents and Multiagent Systems, 2011.

\bibitem[\protect\citeauthoryear{Aziz \bgroup \em et al.\egroup
  }{2013}]{aziz2013computing}
Haris Aziz, Felix Brandt, and Hans~Georg Seedig.
\newblock Computing desirable partitions in additively separable hedonic games.
\newblock {\em Artificial Intelligence}, 195:316--334, 2013.

\bibitem[\protect\citeauthoryear{Banerjee \bgroup \em et al.\egroup
  }{2001}]{banerjee2001core}
Suryapratim Banerjee, Hideo Konishi, and Tayfun S{\"o}nmez.
\newblock Core in a simple coalition formation game.
\newblock {\em Social Choice and Welfare}, 18(1):135--153, 2001.

\bibitem[\protect\citeauthoryear{Bang-Jensen and
  Bessy}{2019}]{bang2019undirected}
J{\o}rgen Bang-Jensen and St{\'e}phane Bessy.
\newblock Degree-constrained 2-partitions of graphs.
\newblock {\em Theoretical Computer Science}, 2019.

\bibitem[\protect\citeauthoryear{Bang-Jensen and
  Christiansen}{2018}]{bang2018directed}
J{\o}rgen Bang-Jensen and Tilde~My Christiansen.
\newblock Degree constrained 2-partitions of semicomplete digraphs.
\newblock {\em Theoretical Computer Science}, 746:112--123, 2018.

\bibitem[\protect\citeauthoryear{Bang-Jensen \bgroup \em et al.\egroup
  }{2016}]{bang2016finding}
J{\o}rgen Bang-Jensen, Nathann Cohen, and Fr{\'e}d{\'e}ric Havet.
\newblock Finding good 2-partitions of digraphs ii. enumerable properties.
\newblock {\em Theoretical Computer Science}, 640:1--19, 2016.

\bibitem[\protect\citeauthoryear{Bazgan \bgroup \em et al.\egroup
  }{2010}]{bazgan2010satisfactory}
Cristina Bazgan, Zsolt Tuza, and Daniel Vanderpooten.
\newblock Satisfactory graph partition, variants, and generalizations.
\newblock {\em European Journal of Operational Research}, 206(2):271--280,
  2010.

\bibitem[\protect\citeauthoryear{Bodlaender}{1994}]{bodlaender1994disjoint}
Hans~L Bodlaender.
\newblock On disjoint cycles.
\newblock {\em International Journal of Foundations of Computer Science},
  5(01):59--68, 1994.

\bibitem[\protect\citeauthoryear{Bogomolnaia and
  Jackson}{2002}]{bogomolnaia2002stability}
Anna Bogomolnaia and Matthew~O Jackson.
\newblock The stability of hedonic coalition structures.
\newblock {\em Games and Economic Behavior}, 38(2):201--230, 2002.

\bibitem[\protect\citeauthoryear{Bulu{\c{c}} \bgroup \em et al.\egroup
  }{2016}]{bulucc2016recent}
Ayd{\i}n Bulu{\c{c}}, Henning Meyerhenke, Ilya Safro, Peter Sanders, and
  Christian Schulz.
\newblock Recent advances in graph partitioning.
\newblock In {\em Algorithm Engineering}, pages 117--158. Springer, 2016.

\bibitem[\protect\citeauthoryear{Chalkiadakis \bgroup \em et al.\egroup
  }{2011}]{chalkiadakis2011computational}
Georgios Chalkiadakis, Edith Elkind, and Michael Wooldridge.
\newblock Computational aspects of cooperative game theory.
\newblock {\em Synthesis Lectures on Artificial Intelligence and Machine
  Learning}, 5(6):1--168, 2011.

\bibitem[\protect\citeauthoryear{Peters}{2016}]{peters2016graphical}
Dominik Peters.
\newblock Graphical hedonic games of bounded treewidth.
\newblock In {\em Thirtieth AAAI Conference on Artificial Intelligence}, 2016.

\bibitem[\protect\citeauthoryear{Shafique and
  Dutton}{2002}]{shafique2002satisfactory}
Khurram~H Shafique and Ronald~D Dutton.
\newblock On satisfactory partitioning of graphs.
\newblock {\em Congressus Numerantium}, pages 183--194, 2002.

\bibitem[\protect\citeauthoryear{Skibski \bgroup \em et al.\egroup
  }{2016}]{skibski2016k}
Oskar Skibski, Szymon Matejczyk, Tomasz~P Michalak, Michael Wooldridge, and
  Makoto Yokoo.
\newblock k-coalitional cooperative games.
\newblock In {\em Proceedings of the 2016 International Conference on
  Autonomous Agents \& Multiagent Systems}, pages 177--185. International
  Foundation for Autonomous Agents and Multiagent Systems, 2016.

\bibitem[\protect\citeauthoryear{Sless \bgroup \em et al.\egroup
  }{2014}]{sless2014forming}
Liat Sless, Noam Hazon, Sarit Kraus, and Michael Wooldridge.
\newblock Forming coalitions and facilitating relationships for completing
  tasks in social networks.
\newblock In {\em Proceedings of the 2014 international conference on
  Autonomous agents and multi-agent systems}, pages 261--268. International
  Foundation for Autonomous Agents and Multiagent Systems, 2014.

\bibitem[\protect\citeauthoryear{Sless \bgroup \em et al.\egroup
  }{2018}]{sless2018forming}
Liat Sless, Noam Hazon, Sarit Kraus, and Michael Wooldridge.
\newblock Forming k coalitions and facilitating relationships in social
  networks.
\newblock {\em Artificial Intelligence}, 259:217--245, 2018.

\bibitem[\protect\citeauthoryear{Stiebitz}{1996}]{stiebitz1996decomposing}
Michael Stiebitz.
\newblock Decomposing graphs under degree constraints.
\newblock {\em Journal of Graph Theory}, 23(3):321--324, 1996.

\end{thebibliography}

\clearpage
\appendix
\section{Charts of graphs}
In the appendix we show only the results of Simulated Annealing and Leximin initiated with partitions from KaHIP, as the ones initiated from random partitions are simply sub performing variations of these.
\begin{figure}[h]
  \centering
  \begin{subfigure}[b]{0.8\linewidth}
    \includegraphics[width=\linewidth]{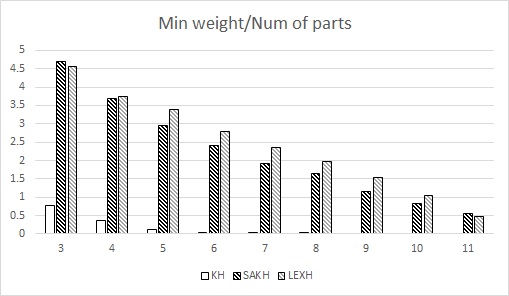}
     \caption{Minimum Weight}
  \end{subfigure}
  \begin{subfigure}[b]{0.8\linewidth}
    \includegraphics[width=\linewidth]{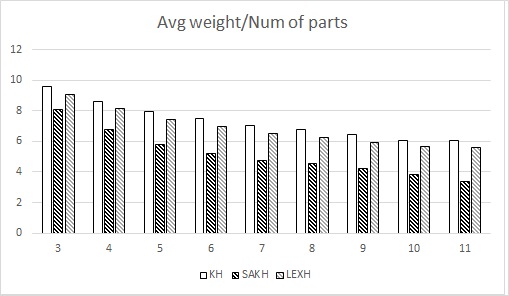}
    \caption{Average Weight}
  \end{subfigure}
  \begin{subfigure}[b]{0.8\linewidth}
    \includegraphics[width=\linewidth]{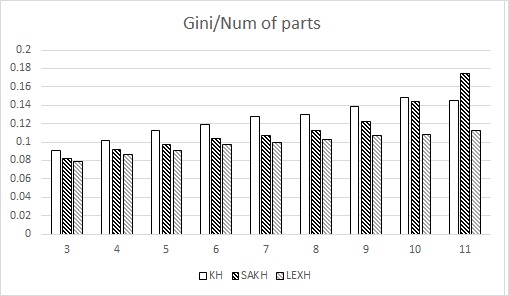}
    \caption{Gini Coefficient}
  \end{subfigure}
  \caption{Performance of algorithms over weighted graphs of size 100, out degree 5, with different number of partitions}
  \label{fig:weighted parts}
\end{figure}

\begin{table}[h]
\begin{adjustbox}{width=0.48\textwidth}
\begin{tabular}{|l|l|l|l|l|}
\hline
                 & \textbf{Min Weight}& \textbf{Avg Weight} & \textbf{Total Weight} &\textbf{Gini} \\ \hline
\textbf{KaHIP}   & 4                & 12.09             & 1318  &0.070        \\ \hline
\textbf{SA\_KH}  & 6                & 11.74             & 1280   &0.074             \\ \hline
\textbf{Lex\_KH} & 6                & 12.01             & 1310  &0.067     \\ \hline
\end{tabular}
\end{adjustbox}

\caption{School 3 Results. 109 Students, 5 friends each, partition to 5 classes.}
\label{school 3}
\end{table}

\begin{figure}
  \centering
  \begin{subfigure}[b]{0.8\linewidth}
    \includegraphics[width=\linewidth]{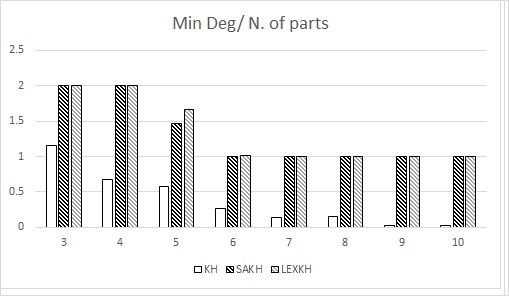}
     \caption{Minimum Degree}
  \end{subfigure}
  \begin{subfigure}[b]{0.8\linewidth}
    \includegraphics[width=\linewidth]{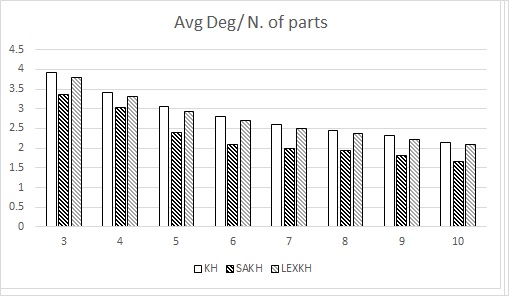}
    \caption{Average Degree}
  \end{subfigure}
  \begin{subfigure}[b]{0.8\linewidth}
    \includegraphics[width=\linewidth]{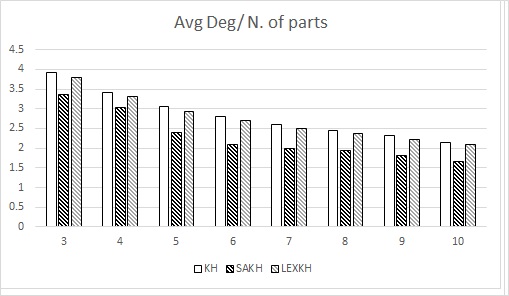}
    \caption{Gini Coefficient}
  \end{subfigure}
  \caption{Performance of algorithms over unweighted graphs of size 100, out degree 7, over different number of partitions}
  \label{fig:unweighted parts}
\end{figure}

\end{document}